\documentclass[11pt]{article}
\usepackage[a4paper,margin=2.5cm]{geometry}
\usepackage{amsmath,amssymb,amsthm,mathtools,bm}
\usepackage{booktabs,array}
\usepackage{graphicx}
\usepackage{tikz}
\usetikzlibrary{arrows.meta,positioning,shapes.geometric,fit}
\usepackage{hyperref}
\usepackage{microtype}
\usepackage{enumitem}
\usepackage{authblk}
\usepackage{siunitx}
\usepackage{xcolor}
\usepackage{float}
\usepackage[right]{lineno}

\hypersetup{colorlinks=true,linkcolor=blue,citecolor=blue,urlcolor=blue}
\newtheorem{proposition}{Proposition}
\newtheorem{principle}{Principle}
\newtheorem{definition}{Definition}
\theoremstyle{remark}
\newtheorem{remark}{Remark}

\title{Earth as a Closed Life-Support System\\
\large Finite Buffers, Dynamical Thresholds, and Planetary Sustainability}
\author[1,2]{Jean-Pierre Gazeau\thanks{Corresponding author. Email: gazeau@apc.in2p3.fr; j.gazeau@uwb.edu.pl}}
\affil[1]{Universit\'e Paris Cit\'e, CNRS, Astroparticule et Cosmologie, F-75013 Paris, France}
\affil[2]{Faculty of Mathematics, University of Bia\l ystok, Bia\l ystok, Poland}
\date{}

\begin{document}
\maketitle




\abstract{
We formulate planetary sustainability as a minimal nonequilibrium problem for coupled finite stocks and fluxes, using the operational language of closed life-support systems: inventories, flux balances, recycling efficiencies, processing capacities, reserve times and failure thresholds. Rather than asserting that waste is universally more limiting than resources, we propose a reserve-time principle: the binding constraint is whichever buffer becomes critical first. A nonlinear model coupling a regenerative resource stock, an accumulated waste stock and an activity variable exhibits a soft threshold, where waste suppresses growth while processing remains formally capable, and a hard threshold, where waste production exceeds maximal processing capacity. With immediate waste feedback, activity settles globally at the soft threshold; the hard threshold is crossed only transiently or under exogenous forcing. A resource-side threshold completes the picture. We extend the model to recycling and to multiple resource and waste classes, give a numerical example, situate the four-function description of life against the origin of life, viruses and aging, contrast the biosphere's two long-term failure modes, and relate the approach to Spaceship Earth, planetary boundaries, Limits to Growth, System Dynamics and stock-pollution economics.}
\vskip 0.3cm


\section{Introduction: the Spaceship Earth analogy, taken seriously}
\begin{quote}
\small\itshape
When the last tree is cut, the last fish is caught, and the last river is polluted … you will realize that you can't eat money.
Alanis Obomsawin (1972)\cite{Obomsawin1972}
\end{quote}

Earth's long-term habitability depends on coupled dynamics of renewable resources, waste production, recycling, and energy throughput. These processes can be viewed as interacting components of a materially bounded life-support system. Apart from small meteoritic inputs and atmospheric escape, Earth is nearly closed to matter, while it receives a large flux of low-entropy solar energy and radiates degraded energy back to space. The metaphor was formulated economically by Boulding~\cite{Boulding1966} and popularized by Fuller~\cite{Fuller1969}. The stock-and-flow language used here follows Forrester's System Dynamics~\cite{Forrester1961,Forrester1971}, applied to global sustainability by Meadows et al.~\cite{Meadows1972}; the present treatment differs in using continuous nonlinear dynamics with explicit analytical thresholds rather than simulation-based causal-loop modelling.

The analogy is usually invoked rhetorically. Our purpose is operational: to ask whether Earth can be studied with concepts already used for spacecraft, orbital stations, and regenerative life-support projects, which are designed and audited through explicit accounts of stocks and fluxes and survive only if several coupled loops close stably at once.

This viewpoint has a long economic pedigree. Mill's stationary state~\cite{Mill2011} and Daly's steady-state economy~\cite{Daly1977} question indefinite expansion of material throughput, and Georgescu-Roegen~\cite{GeorgescuRoegen1971} emphasized the entropic asymmetry between material reuse and energy degradation. The present note translates part of that intuition into a deliberately minimal dynamical framework.

The central question is not whether ``resources'' or ``waste'' are universally more important, but:
\begin{quote}
\emph{Which finite buffer becomes critical first: the stock from which activity draws, or the capacity by which its residuals are transformed, exported, diluted, stored, or isolated?}
\end{quote}
Food deprivation can be tolerated for weeks, while complete renal failure without replacement therapy can become fatal within days to weeks; yet the loss of oxygen or water, both resources, is faster still. The distinction lies not in the labels but in the ratio between available buffer and throughput.

The model below isolates the waste-elimination mechanism because it yields two transparent thresholds, while retaining a resource-side threshold and a symmetric resource-side reserve time. It is not a calibrated Earth-system model; its purpose is to identify measurable quantities and failure modes that a detailed, data-assimilated model should retain. 
From a statistical-physics viewpoint, the model belongs to a broad family of driven nonequilibrium
systems in which throughput competes with finite regenerative and dissipative capacities, potentially
giving rise to distinct dynamical regimes separated by critical thresholds \cite{PrigogineStengers1984,SmithWaltman1995,Scheffer2009}.

\section{Earth as a closed life-support system}

Let $X(t)=(X_1,\ldots,X_m)$ denote material inventories such as freshwater, nutrients, soil, biomass, minerals, and atmospheric constituents. Their generic balance takes the form
\begin{equation}
\dot X_i = J_i^{\mathrm{in}}-J_i^{\mathrm{out}}+J_i^{\mathrm{regen}}-J_i^{\mathrm{use}}+J_i^{\mathrm{rec}}.
\label{eq:balance}
\end{equation}
Natural regeneration and organized recycling are distinct, finite processes. Energy is different: Earth receives radiation, converts part of it into chemical, kinetic, and thermal forms, and exports entropy through infrared radiation. Matter can circulate; usable energy cannot be recycled indefinitely. A life-support audit therefore requires the quantities listed in Table~\ref{tab:audit}.

\begin{table}[ht]
\centering
\caption{Minimal audit variables for a closed life-support system.}
\label{tab:audit}
\begin{tabular}{>{\raggedright\arraybackslash}p{0.27\linewidth}p{0.66\linewidth}}
\toprule
Quantity & Operational meaning \\
\midrule
Stocks & Available resources and accumulated wastes at a given time. \\
Consumption fluxes & Resource use per unit time, by sector or activity class. \\
Regeneration fluxes & Natural replenishment rates and their environmental dependence. \\
Waste-production fluxes & Residual material generated per unit activity and time. \\
Processing capacities & Maximum rates at which waste can be transformed, exported, diluted, stored, or rendered harmless. \\
Recycling efficiencies & Fractions of material streams returned to useful stocks. \\
Reserve times & Time to reach a critical stock, concentration, or capacity margin under a stated scenario. \\
Margins and thresholds & Distances to boundaries beyond which recovery, regulation, or survival fails. \\
\bottomrule
\end{tabular}
\end{table}

The International Space Station operates with explicit budgets for oxygen, water, food, and recycling. Biosphere~2 showed that a sealed ecological system can fail through internal feedbacks not anticipated from resource inventories alone~\cite{Nelson1996}. ESA's MELiSSA programme develops biologically regenerative loops around coupled resource-and-waste budgets. These are imperfect analogues of Earth\footnote{Unlike an orbital space station, which can export waste beyond its physical boundaries, Spaceship Earth has essentially no external material sink. Waste must therefore be stored, transformed, recycled, or assimilated within the finite capacities of the terrestrial system.}, but they establish a methodological point: viability is a network property.

\section{Life as organized throughput: from cell to organism to planet}
\label{sec:throughput}

Following Schr\"odinger~\cite{Schrodinger1944}, a living system maintains organization by exporting entropy to its environment. Prigogine~\cite{Prigogine1977} generalized this view through dissipative structures, and Vernadsky~\cite{Vernadsky1998} framed the biosphere itself as a self-organizing thermodynamic envelope.

Here life is described operationally through four jointly necessary functions: (i) controlled capture and conversion of energy and material from the environment; (ii) preservation and transmission of information; (iii) regulation and adaptation; and (iv) transformation, recycling, sequestration, or elimination of harmful residuals. None is sufficient alone.

Function~(i) is the dissipative-structure picture~\cite{Prigogine1977}, extended to ecosystems~\cite{SchneiderKay1994,UlanowiczHannon1987} and formalized through entropy-production principles~\cite{Kleidon2010}; it is indifferent to whether the exploited gradient is chemical, thermal, or radiative~\cite{Michaelian2011}, which is why the same language transfers from metabolism to planetary budgets. Function~(ii) requires a material memory causally decoupled from any single physical instantiation and copied with heritable variation~\cite{WalkerDavies2013}. Function~(iii) is consistent with, but not established by, the tendency of strongly driven systems to favour configurations that dissipate free energy efficiently~\cite{England2013,England2015}, a bias already present in non-living systems. Together, functions~(ii) and~(iii) amount to what Beniger called \emph{control}: programmed information and feedback that any throughput-processing system requires, from the cell to industrial society, whose throughput repeatedly outran its means of regulation~\cite{Beniger1986}. Function~(iv) is developed quantitatively in Sections~\ref{sec:model}--\ref{sec:loops}.

\begin{remark}[Scope and limits of the four-function description]
\label{rem:scope}
This description characterizes life as an organized thermodynamic and informational throughput process. It does not distinguish whether throughput is carried by an organism, a population, or the biosphere (Fig.~\ref{fig:scales}): the aggregate variable $N$ of Section~\ref{sec:model} does not resolve an individual from its population. This has a checkable cost. Self-sustaining autocatalytic networks that recycle their own inventory without a decoupled template lack Darwinian evolvability: replication is too inaccurate for selection to retain fitter variants~\cite{VasasSzathmarySantos2010}. Recycling and waste-processing capacity are therefore a necessary, largely universal feature of living systems rather than a sufficient characterization. Replication-centred definitions --- NASA's self-sustaining chemical system capable of Darwinian evolution~\cite{Joyce1994}, or Pross's kinetic selection among self-replicating networks~\cite{Pross2011,Pross2016} --- make heritable variation the decisive criterion, a role function~(ii) plays only nominally in aggregate stock-and-flow language. The framework is thus a thermodynamic accounting layer, not a substitute for an evolutionary account of life.
\end{remark}

\begin{figure}[H]
\centering
\begin{tikzpicture}[
  node distance=12mm and 13mm,
  >=Latex,
  box/.style={draw, rounded corners=2mm, align=center, text width=3.0cm, minimum height=1.35cm, inner sep=4pt, font=\small},
  small/.style={draw, rounded corners=1.5mm, align=center, text width=3.7cm, minimum height=1.05cm, inner sep=4pt, font=\small},
  arrow/.style={->, thick}
]
\node[box] (cell) {\textbf{Cell}\\metabolic inputs\\degradation, recycling, export};
\node[box, right=of cell] (org) {\textbf{Organism}\\food, water,\\oxygen\\liver, kidneys, lungs, excretion};
\node[box, right=of org] (planet) {\textbf{Planet}\\resource stocks and regeneration\\sinks, treatment, recycling, isolation};

\draw[arrow] (cell) -- (org);
\draw[arrow] (org) -- (planet);
\node[font=\scriptsize, above=4mm of org] {increasing organizational scale $\longrightarrow$};

\node[small, below=13mm of org] (principle) {Common requirement:\\\textbf{finite buffers + finite processing}\\must match throughput};
\draw[arrow] (cell.south) |- (principle.west);
\draw[arrow] (org.south) -- (principle.north);
\draw[arrow] (planet.south) |- (principle.east);
\end{tikzpicture}
\caption{A multiscale life-support analogy. No single cellular organelle is a ``kidney'': degradation, recycling, transport, detoxification, and export are distributed functions. At organismal and planetary scales the implementation changes, but viability still depends on matching throughput to finite stocks and finite processing capacities.}
\label{fig:scales}
\end{figure}

\subsection{Anthropic time and the probability of a singular origin event}
\label{sec:anthropic}

The transition from the abiotic Hadean--Eoarchean environment to an organization capable of memory, adaptation, and waste management happened, as far as the terrestrial record shows, once, and within a finite, astronomically imposed time budget. The habitable window of a Sun-like star is $\tau_{\mathrm{hab}}\sim\mathcal{O}(10)$~Gyr, and life appears within a few hundred million years after the end of the Hadean bombardment~\cite{SpiegelTurner2012}.

Reading this early timing as evidence that abiogenesis is chemically easy is not, in general, valid~\cite{Carter1983}. If some steps toward observers have expected waiting times exceeding the habitable window, observers can only find themselves in the rare histories in which those steps happened early enough anyway: the observed timing is then a selection effect, not a rate measurement. Carter's argument, originally about intelligence, applies equally to the origin of life~\cite{Watson2008}. Spiegel and Turner showed in a Bayesian treatment that the single terrestrial data point leaves the posterior on the intrinsic abiogenesis rate dominated by the prior, compatible with both near-certain and astronomically improbable origins~\cite{SpiegelTurner2012}. Kipping's joint analysis of life and intelligence narrows this uncertainty without removing it~\cite{Kipping2020}, and Snyder-Beattie et al.\ extend the hard-step framework to the major evolutionary transitions~\cite{SnyderBeattie2021}.

For this note, the implication concerns how to read the single data point, not life elsewhere. What is not in question is the finiteness of the window: whatever its probability, the transition had to occur within an interval set by astrophysical constraints external to biochemistry --- the same structural feature, on a vastly different timescale, that motivates the present concern with finite buffers and bounded reserve times.

\subsection{Viruses: a limit case of the four-function description}
\label{sec:viruses}

\begin{definition}[Virus]
Following Raoult and Forterre's genomic reformulation of Lwoff's criteria, a virus is a capsid-encoding organism: an infectious genetic element of nucleic acid and protein that self-assembles into a nucleocapsid and completes its replication only by redirecting the ribosomes of a ribosome-encoding organism, i.e.\ a cell~\cite{RaoultForterre2008}. Between infections a virion has no autonomous metabolism.
\end{definition}

A virus fails function~(i): it captures no energy or material of its own and appropriates the host's organized throughput. It satisfies~(ii), its genome being a heritable memory copied independently of any single virion, and~(iii) emphatically: RNA viruses mutate and are selected far faster than cellular life because they free-ride on host replication machinery. Function~(iv) applies only derivatively, through the residuals released by lysis. Removing~(i) alone thus leaves replicating, adapting, informational organization without autonomous throughput --- which is why the status of viruses as living remains unsettled, particularly since the discovery of giant viruses rivalling small bacteria~\cite{RaoultForterre2008}.

Viral hallmark genes, shared across viruses infecting all three domains but absent from cellular genomes, suggest that such replicators are as ancient as cellular diversification itself~\cite{KooninSenkevichDolja2006}: any sufficiently rich replicating chemistry creates a niche for variants that copy the replication machinery without paying its resource cost. This attack--adaptation--selection loop is observable today. CRISPR-Cas lets bacteria and archaea store fragments of past viral genomes as heritable spacers used to cleave later infections~\cite{Barrangou2007}: function~(ii) evolving against a parasite lacking~(i), a molecular instance of Red Queen coevolution~\cite{VanValen1973}. At ocean scale, the viral shunt returns a substantial fraction of primary production to the dissolved organic pool, making viruses major drivers of planetary recycling~\cite{Suttle2007}.

\subsection{From distributed processing to reserve times}
\label{sec:reserve}

The same logic appears across scales (Fig.~\ref{fig:scales}). A cell has no unique ``renal'' organelle: proteasomal and lysosomal degradation, autophagy, membrane transport, redox regulation, and metabolite export collectively keep internal residuals from overwhelming homeostasis. Organisms use specialized organs; the planet uses biogeochemical sinks, ecosystem transformation, technical treatment, recycling, storage, and isolation. The analogy is functional rather than anatomical, and fasting versus renal failure illustrates different reserve times, not a universal precedence of waste over resources.

\begin{principle}[Reserve-time limited viability]
In a materially bounded life-support system, long-term viability requires both (i) resource stocks and regeneration rates compatible with demand and (ii) residual-production rates compatible with finite processing, storage, dilution, export, or isolation capacities. The binding constraint is the one whose associated reserve time or capacity margin becomes critical first.
\end{principle}

Persistent waste may also damage the processes that generate resources or maintain regulation, so the two constraints are often coupled. If $\eta\in[0,1]$ is the recovery fraction of a material stream, repeated cycling leaves losses unless $\eta=1$, and recovery consumes energy and can create secondary residuals: recycling delays thresholds but does not abolish them.

The principle also situates \emph{aging}. Biological aging is increasingly described as progressive accumulation of unrepaired molecular and cellular damage once maintenance and clearance fail to keep pace~\cite{LopezOtin2013}; in the notation of Section~\ref{sec:model} this is a slow endogenous drift of $b$ upward or $\varepsilon_{\max}$ downward, pushing a nominally stationary system past its own thresholds. Amorphous glasses and polymers show an analogous \emph{physical aging}, a slow structural relaxation with drifting properties and no repair involved~\cite{Struik1978}. Aging can thus be read as a slowly driven regime of Principle~1 in which the buffers or processing capacities themselves erode.

\section{A minimal nonequilibrium dynamical model}
\label{sec:model}

Let $R(t)\ge 0$ be an effective regenerative resource stock, $W(t)\ge 0$ an effective accumulated harmful-waste stock, and $N(t)\ge 0$ an aggregate measure of population, economic activity, or organized throughput. Consider
\begin{align}
\dot R &= \rho(R)-aN+\chi\,\varepsilon(W), \label{eq:R}\\
\dot W &= bN-\varepsilon(W), \label{eq:W}\\
\dot N &= N\,g(R,W). \label{eq:N}
\end{align}
Here $a>0$ is resource use per unit activity, $b>0$ gross waste production per unit activity, $\rho(R)$ natural regeneration, $\varepsilon(W)$ a saturating waste-processing flux, and $\chi\in[0,1]$ the fraction of processed waste returned to the resource pool. For concreteness, take
\begin{equation}
\rho(R)=sR\left(1-\frac{R}{R_{\max}}\right),
\qquad
\varepsilon(W)=\varepsilon_{\max}\frac{W}{K_W+W},
\label{eq:functions}
\end{equation}
where $K_W>0$ is the half-saturation constant of processing (not a carrying capacity), and first isolate the waste-driven mechanism by setting $\chi=0$ and $g(R,W)\to g(W)$ with
\begin{equation}
g(W)=r\left(1-\frac{W}{W_c}\right),
\label{eq:g}
\end{equation}
where $W_c>0$ is the waste level at which net activity growth changes sign.

\subsection{Soft and hard waste thresholds}

The waste balance $\dot W=0$ defines a nullcline\footnote{A nullcline is a curve in the phase plane along which one state variable has zero rate of change; equilibria lie at the intersections of nullclines.}, not an equilibrium of the full system. Since $\varepsilon(W)<\varepsilon_{\max}$, it exists only for activity levels with $bN<\varepsilon_{\max}$, where it is the increasing curve
\begin{equation}
W_0(N)=\frac{bNK_W}{\varepsilon_{\max}-bN}.
\label{eq:Wstar}
\end{equation}
The activity nullcline (for $N>0$) is the vertical line $W=W_c$.

\begin{proposition}[Soft and hard waste thresholds]
\label{prop:thresholds}
For the subsystem $\dot W=bN-\varepsilon(W)$, $\dot N=Nr(1-W/W_c)$ on $W\ge0$, $N>0$, define
\begin{equation}
N_*:=\frac{\varepsilon(W_c)}{b}=\frac{W_c\,\varepsilon_{\max}}{b(K_W+W_c)},
\qquad N_c:=\frac{\varepsilon_{\max}}{b},
\qquad\text{so that } 0<N_*<N_c .
\label{eq:thresholds}
\end{equation}
\begin{enumerate}[label=(\roman*)]
\item \emph{Soft threshold.} Along the waste nullcline, $\dot N>0$ for $N<N_*$ and $\dot N<0$ for $N_*<N<N_c$: between $N_*$ and $N_c$, the waste stock that processing can balance already exceeds the level compatible with growth.
\item \emph{Hard threshold.} If $N\ge N_c$, then $\dot W>0$ for every finite $W$: production exceeds maximal processing capacity.
\item \emph{Self-regulation.} $(W_c,N_*)$ is the unique equilibrium with $N>0$, and it attracts every trajectory with $N(0)>0$. It is a stable focus when $\varepsilon'(W_c)^2<4b\,rN_*/W_c$.
\end{enumerate}
\end{proposition}

\noindent\emph{Proof.} (i) $W_0$ is increasing and $W_0(N_*)=W_c$. (ii) $bN\ge\varepsilon_{\max}>\varepsilon(W)$. (iii) The domain $D=\{W\ge0,\,N>0\}$ is positively invariant, since $N$ obeys $\dot N=Ng$ and $\dot W=bN>0$ at $W=0$. The Volterra-type function
\[
V(W,N)=\frac{r}{2bW_c}(W-W_c)^2+N-N_*-N_*\ln\frac{N}{N_*}
\]
has compact sublevel sets in $D$, and, the cross terms cancelling exactly,
\[
\dot V=\frac{r}{bW_c}\,(W-W_c)\bigl(\varepsilon(W_c)-\varepsilon(W)\bigr)\le0,
\]
because $\varepsilon$ is increasing. By LaSalle's invariance principle\footnote{LaSalle's invariance principle~\cite{LaSalle1960}: if $V$ is continuously differentiable on a positively invariant set $D$, has compact sublevel sets there, and satisfies $\dot V\le0$ along solutions, then every solution starting in $D$ approaches the largest invariant subset of $\{\dot V=0\}$. Unlike Lyapunov's direct method, it does not require $\dot V<0$ away from the equilibrium.}, trajectories approach the largest invariant subset of $\{\dot V=0\}=\{W=W_c\}$. A solution staying on that line has $\dot W=0$, hence $bN=\varepsilon(W_c)$ and $N=N_*$. Finally, the Jacobian at $(W_c,N_*)$ has trace $-\varepsilon'(W_c)<0$ and determinant $brN_*/W_c>0$, whence the focus condition. \hfill$\square$

\begin{figure}[htbp]
\centering
\includegraphics[width=0.9\textwidth]{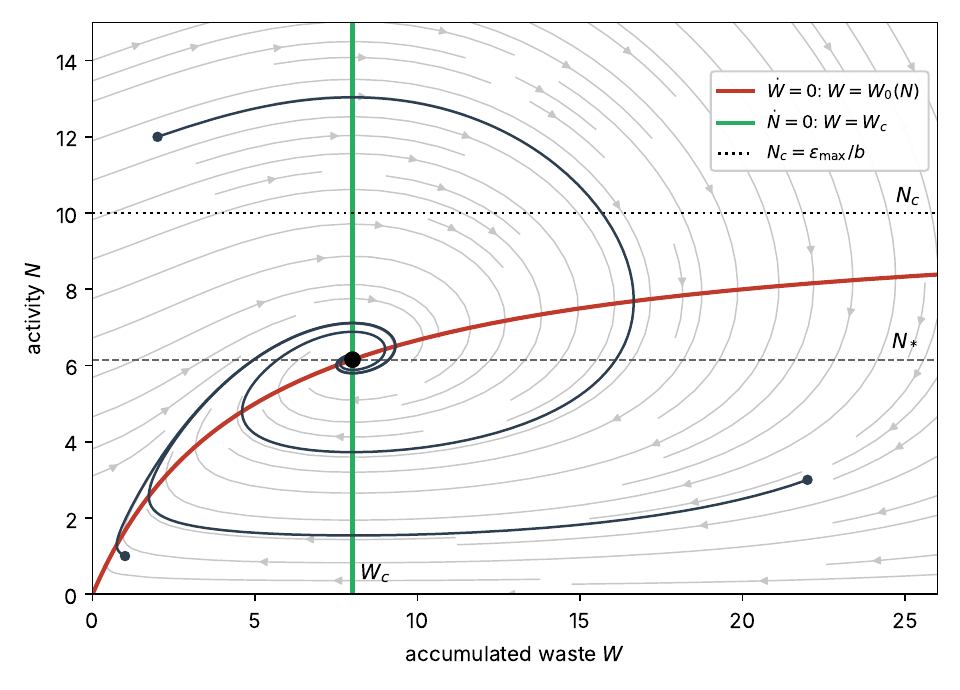}
\caption{Phase portrait of the waste--activity subsystem for the parameters of Section~\ref{sec:toy} ($b=1$, $\varepsilon_{\max}=10$, $K_W=5$, $W_c=8$) and $r=0.3$. Red: waste nullcline $W=W_0(N)$; green: activity nullcline $W=W_c$; dotted: hard threshold $N_c$; dashed: soft threshold $N_*$. All trajectories with $N>0$ spiral into $(W_c,N_*)$; the trajectory starting above $N_c$ returns below it.}
\label{fig:phase}
\end{figure}

The two thresholds thus play different roles (Fig.~\ref{fig:phase}). With immediate waste feedback on growth, activity regulates itself at the soft threshold $N_*$, below the processing limit; the hard threshold $N_c$ can only be crossed transiently, as overshoot. It becomes a persistent regime only when activity is driven independently of accumulating waste: when growth does not respond to $W$, when the response is delayed, or when $N$ is set exogenously. Beniger's historical ``crisis of control'', in which throughput outpaces the feedback that regulates it, is the social form of this regime~\cite{Beniger1986}. The model therefore separates two ways a society can relate to its residuals: being slowed by them, or outrunning its capacity to process them.

\subsection{Reserve times}

Proposition~\ref{prop:thresholds} concerns the waste side and does not prove that waste must dominate resource scarcity. To compare the two sides, introduce reserve times. Let activity be held at a level $N>N_c$ and let $W(0)=W_{\mathrm i}<W_c$. Since $\varepsilon(W)\le\varepsilon_{\max}$, the time $\tau_W$ needed to reach $W_c$ obeys
\begin{equation}
\tau_W\le\frac{W_c-W_{\mathrm i}}{bN-\varepsilon_{\max}},
\label{eq:tauW}
\end{equation}
with near-equality when $W\gg K_W$. Integrating Eq.~\eqref{eq:W} at fixed $N$ gives the exact value
\begin{equation}
\tau_W=\frac{W_c-W_{\mathrm i}}{\Delta}-\frac{K_W\varepsilon_{\max}}{\Delta^2}\ln\frac{\Delta W_c+bNK_W}{\Delta W_{\mathrm i}+bNK_W},
\qquad \Delta:=bN-\varepsilon_{\max}.
\label{eq:tauWexact}
\end{equation}
Likewise, if demand exceeds regeneration, a resource stock starting at $R_{\mathrm i}$ with critical level $R_{\mathrm{crit}}$ has the approximate depletion time
\begin{equation}
\tau_R\simeq \frac{R_{\mathrm i}-R_{\mathrm{crit}}}{aN-\rho(R)}.
\label{eq:tauR}
\end{equation}
The two expressions share the same balance structure: a short $\tau_R$ makes a resource behave like a hard constraint, and a short $\tau_W$ signals rapid loss of viability from residual accumulation.

\begin{remark}[Competition of buffers]
The pair $(\tau_R,\tau_W)$ generalizes the physiological analogy. Food reserves correspond to a long $\tau_R$, oxygen to an extremely short one; renal failure illustrates a small waste-processing margin, not a theorem that waste is always more dangerous than resource loss. At any scale, the operative question is which buffer reaches its critical boundary first.
\end{remark}

\subsection{The resource side and its coupling to waste}

Returning to the full system, a resource equilibrium requires
\begin{equation}
aN=\rho(R^*)+\chi\varepsilon(W^*)=\rho(R^*)+\chi bN,
\label{eq:Req}
\end{equation}
the second equality holding at waste balance. Since $\rho\le sR_{\max}/4$, no resource equilibrium exists, whatever the waste state, once $N$ exceeds the resource-side threshold
\begin{equation}
N_R:=\frac{sR_{\max}}{4(a-\chi b)}\qquad(a>\chi b).
\label{eq:NR}
\end{equation}
Which of $N_*$, $N_c$ and $N_R$ is smallest is a matter of parameters, not of principle: this is Principle~1 in threshold form. The two sides are moreover coupled when pollution reduces regeneration,
\begin{equation}
\rho(R,W)=sR\left(1-\frac{R}{R_{\max}}\right)q(W),\qquad q'(W)<0,
\label{eq:rhoW}
\end{equation}
which lowers $N_R$ as $W$ grows. If, in addition, processing capacity depends on available resources, $\varepsilon_{\max}=\varepsilon_{\max}(R)$, a reinforcing loop closes: waste reduces regeneration, declining resources impair treatment, and residuals accumulate faster.

\subsection{Recycling and technological change}
\label{sec:recycling}

Suppose a fraction $\eta\in[0,1)$ of the gross residual stream is recovered before entering the harmful stock and returned to the resource pool, so that $\dot W=(1-\eta)bN-\varepsilon(W)$ and $\eta bN$ is added to $\dot R$. Then
\begin{equation}
N_c(\eta)=\frac{\varepsilon_{\max}}{(1-\eta)b},
\qquad
N_*(\eta)=\frac{W_c\,\varepsilon_{\max}}{(1-\eta)b(K_W+W_c)},
\label{eq:recycling}
\end{equation}
and $N_R$ increases as well. Recycling raises both waste thresholds by the common factor $1/(1-\eta)$ but removes them only in the idealized limit $\eta\to1$. Efficiency gains in $a$ or $b$ can be offset by growth in scale $N$; the system-level quantities are $aN$ and $(1-\eta)bN$. Processing capacity $\varepsilon_{\max}$ is itself dynamic: investment and innovation can raise it, while ageing infrastructure, conflict, energy scarcity, or pollution damage can lower it. Raising it is not free. Tainter argues that societies solve problems by investing in complexity whose marginal returns decline until its maintenance costs more than it yields~\cite{Tainter1988}. In the present model, maintaining $\varepsilon_{\max}$ draws a resource flux $c(\varepsilon_{\max})$, with $c$ convex, from $\dot R$; this lowers $N_R$ to $\bigl(sR_{\max}/4-c(\varepsilon_{\max})\bigr)/(a-\chi b)$, so gains on the waste side are paid for on the resource side.

\section{Multiple resources, multiple wastes, and planetary loops}
\label{sec:loops}

A single $R$ and $W$ is analytically useful but empirically insufficient. Writing $\bm R=(R_1,\ldots,R_m)$ and $\bm W=(W_1,\ldots,W_q)$, one may generalize to
\begin{align}
\dot{\bm R}&=\bm\rho(\bm R,\bm W)-A\bm n+C\bm\varepsilon(\bm W),\\
\dot{\bm W}&=B\bm n-\bm\varepsilon(\bm W),\\
\dot{\bm n}&=\operatorname{diag}(\bm n)\,\bm g(\bm R,\bm W).
\label{eq:vector}
\end{align}
Real categories differ not only by parameter values but by the form and timescale of their regeneration and processing functions.

\paragraph{Carbon.}
Fossil carbon formed over geological timescales and is extracted over industrial ones; its regeneration rate is effectively zero over any relevant window. Combustion returns carbon to the atmosphere faster than oceans, vegetation, and soils can reabsorb it~\cite{IPCC2021,Smil2017}; the elimination function is finite and saturable.

\paragraph{Water.}
Water is not destroyed by use, but its location and quality change. Potable water requires treatment, energy, and ecosystem services, so one must distinguish usable and contaminated stocks with a treatment loop between them. Some regional systems possess only short buffers even though the global hydrological cycle is large.

\paragraph{Nuclear waste.}
High-level radioactive waste is reduced mainly by radioactive decay. Deep geological disposal chiefly isolates rather than processes the stock~\cite{IAEA2009}, so on institutional timescales the effective elimination rate may be extremely small.

\paragraph{Chemical and agricultural residuals.}
Fertilizers, pesticides, and livestock effluents affect soil and freshwater continuously~\cite{Steinfeld2006}. Nitrogen and phosphorus flows are central to the planetary-boundaries literature~\cite{Rockstrom2009a,Rockstrom2009b,Steffen2015}. Here contamination acts directly on regeneration, so resource decline and residual accumulation become two aspects of one degradation process.

\paragraph{Pharmaceutical and personal-care residuals.}
Medicines and cosmetics form a distinct class of residuals, long recognized as ``agents of subtle change''~\cite{DaughtonTernes1999}. They reach the environment continuously and diffusely, through excretion, wash-off and disposal, so their emission scales with activity and is well described by the term $bN$. A survey of more than a thousand river sites on all continents found at least one active pharmaceutical ingredient above concentrations considered safe for aquatic organisms at about a quarter of them~\cite{Wilkinson2022}. Three features distinguish this class in the present framework. First, the threshold $W_c$ is unusually low: these molecules are \emph{designed} to be biologically active at small doses. A whole-lake experiment showed that a few nanograms per litre of the synthetic estrogen of contraceptive pills caused the near-collapse of a fathead minnow population~\cite{Kidd2007}, and veterinary diclofenac in livestock carcasses drove a decline of more than 95\% of several South Asian vulture species~\cite{Oaks2004}. Both are direct instances of the coupling $q(W)$ of Eq.~\eqref{eq:rhoW}: the residual acts on regeneration, not merely on a stock. UV filters from sunscreens have likewise been implicated in coral bleaching~\cite{Danovaro2008}. Second, $\varepsilon_{\max}$ is small: many of these compounds are designed for stability, and conventional wastewater treatment was not built to remove them. Third, and more unusually, antibiotic residues can degrade the processing capacity itself, since they act on the microbial communities that perform much of the natural elimination, so that $\varepsilon_{\max}$ decreases with $W$; and the associated resistance genes are residuals that \emph{replicate}, spreading by selection and horizontal transfer~\cite{LarssonFlach2022}. Such a stock acquires its own growth term, beyond the linear bookkeeping of Eq.~\eqref{eq:W}. Finally, the number of distinct synthetic compounds is so large that the vector $\bm W$ is mostly unmonitored: the planetary boundary for ``novel entities'' has been assessed as already transgressed, in part because production outpaces the capacity for safety assessment~\cite{Persson2022}. In the language of Section~\ref{sec:model}, the regulatory and monitoring capacity is itself a finite processing function that can be outrun.

\paragraph{A historical precedent: Amazonian dark earths.}
Pre-Columbian societies of central Amazonia incorporated charcoal, ash, bone, and organic refuse into nutrient-poor tropical soils over centuries, producing \emph{terra preta de \'Indio}: soils with several times the organic matter, nitrogen, and phosphorus of surrounding Ferralsols, whose fertility persists long after abandonment~\cite{Glaser2007}. In the present language this practice raised the recovery fraction $\eta$ of a domestic and agricultural waste stream, and with it both $N_c(\eta)$ and $N_*(\eta)$ in Eq.~\eqref{eq:recycling}: a documented, pre-industrial demonstration that sustained high-$\eta$ waste management is achievable.

These categories should not be collapsed into a single moral or physical category. Their commonality is formal: production, transformation, transport, storage, recycling, and finite rates of removal or isolation. The empirical task is to identify the appropriate state variables and timescales for each class.

\section{A toy explicit numerical example}
\label{sec:toy}

The purpose of this section is pedagogical, not predictive. Take
\begin{equation}
b=1,\qquad \varepsilon_{\max}=10,\qquad K_W=5,\qquad W_c=8,
\end{equation}
so that
\begin{equation}
N_c=10,\qquad N_*=\frac{80}{13}\simeq 6.15.
\end{equation}
First hold activity at a fixed, exogenous level $N$, as for a society whose growth does not yet respond to its residuals. Waste then relaxes to $W_0(N)$ when $N<N_c$ and grows without bound when $N\ge N_c$ (Table~\ref{tab:toy}).

\begin{table}[ht]
\centering
\caption{Toy example at fixed activity $N$, with $W_0(N)=5N/(10-N)$.}
\label{tab:toy}
\begin{tabular}{ccl}
\toprule
$N$ & $W_0(N)$ & sign of $g(W_0)$ \\
\midrule
3 & $\approx2.14$ & positive \\
6 & $7.50$ & positive, near arrest \\
$80/13\simeq6.15$ & $8.00$ & zero (soft threshold) \\
8 & $20.0$ & negative \\
9 & $45.0$ & negative \\
10 & none (unbounded) & hard threshold \\
\bottomrule
\end{tabular}
\end{table}

The striking feature is the separation of thresholds: the waste level can remain finite while already lying beyond the level compatible with growth. Once feedback acts, Proposition~\ref{prop:thresholds} applies: for $r=0.3$ the eigenvalues at $(W_c,N_*)$ are $\simeq-0.15\pm0.46\,\mathrm{i}$, and every trajectory spirals into $N_*$ through damped oscillations of period $\simeq14$ time units (Fig.~\ref{fig:phase}).

Suppose instead that activity is held at $N=12>N_c$ with $W_{\mathrm i}=2$. The bound~\eqref{eq:tauW} gives $\tau_W\le(8-2)/(12-10)=3$, whereas the exact value~\eqref{eq:tauWexact} is $\tau_W\simeq0.85$: while $W$ is not large compared with $K_W$, processing runs well below $\varepsilon_{\max}$, and the threshold is reached more than three times sooner. Reserve-time bounds are diagnostic, not substitutes for integrating the dynamics near a threshold.

Now take $\eta=0.4$. Equation~\eqref{eq:recycling} rescales both waste thresholds by $1/(1-\eta)=5/3$:
\begin{equation}
N_c(0.4)\simeq16.67,
\qquad
N_*(0.4)\simeq10.26.
\end{equation}
Recycling can therefore change the qualitative operating regime without eliminating finite thresholds.

\section{Carrying capacity as a dynamical quantity}
\label{sec:capacity}

Classical carrying capacity is often a constant in a logistic equation. In a coupled life-support system this is too restrictive.

\begin{definition}[Dynamical carrying capacity]
Given a model, an admissible region $\mathcal A$ in state space, its viability kernel $\operatorname{Viab}(\mathcal A)$, a time horizon $T$, and admissible controls (activity level, recycling fraction $\eta$, processing capacity $\varepsilon_{\max}$, intensities $a$ and $b$), the carrying capacity is the largest sustained activity scale for which trajectories remain in $\mathcal A$ over $[0,T]$ and end in $\operatorname{Viab}(\mathcal A)$ at $t=T$, within an accepted risk tolerance.
\end{definition}

Here $\operatorname{Viab}(\mathcal A)$ is the viability kernel of Aubin's viability theory, the set of states from which some admissible control keeps the evolution inside $\mathcal A$ for all future times; that theory was developed in part to formalize questions of sustainability~\cite{Aubin1991,AubinBayenSaintPierre2011}. The terminal condition is essential: without it, a state may lie in $\mathcal A$ at $t=T$ yet be doomed to leave it under every future control. In the toy example, activity held above $N_c$ keeps $W<W_c$ for a time $\tau_W$, and would be declared admissible on any horizon $T<\tau_W$.

Carrying capacity thus depends on time horizon and risk tolerance; it changes with technology, institutions, ecosystem state, and distribution; it may decline through hysteresis; and it is generally uncertain. A system can be below one threshold and above another: resource adequacy does not imply waste compatibility, and high processing capacity does not imply resource security. A useful diagnostic is therefore a set of margins and reserve times,
\begin{equation}
\mathcal M(t)=\{\tau_{R_i}(t),\,\tau_{W_j}(t),\,\Delta_i(t),\,\Delta_j(t)\},
\end{equation}
where $\Delta$ denotes distance to a relevant threshold. The smallest element need not remain associated with the same variable as technology, behaviour, or environmental conditions evolve.

\section{Two long-term fates: resource exhaustion versus self-poisoning}
\label{sec:fates}

Which of the model's two kinds of failure, resource-side or waste-side, governs the eventual, non-catastrophic end of Earth's biosphere, on timescales over which $R_{\max}$ and the parameters of $\varepsilon(W)$ themselves drift? Both candidate mechanisms are attested in the Earth-system literature.

\paragraph{Resource exhaustion (an $R$-side mechanism).} As the Sun's luminosity rises, silicate weathering accelerates and draws atmospheric CO$_2$ toward the compensation point of photosynthesis. Caldeira and Kasting first quantified the resulting finite lifespan of the biosphere~\cite{CaldeiraKasting1992}; Franck, Bounama and von Bloh project that complex multicellular life, eukaryotes, and prokaryotes die out in that order --- the reverse of their appearance --- within roughly 0.8 to 1.6~Gyr~\cite{FranckBounamaVonBloh2006}, the biosphere retreating into microbial refugia~\cite{OMalleyJames2013}. In the language of Section~\ref{sec:capacity}, this is a Gyr-scale erosion of carrying capacity: $R_{\max}(t)$ and $\rho(R)$ decline under an external forcing that Eqs.~\eqref{eq:R}--\eqref{eq:N} do not themselves generate, lowering the resource-side threshold $N_R$ of Eq.~\eqref{eq:NR}.

\paragraph{Self-poisoning (a $W$-side, hard-threshold mechanism).} Kump, Pavlov and Arthur show that once dissolved H$_2$S in an anoxic deep ocean exceeds a critical concentration, the chemocline can rise abruptly to the surface, releasing H$_2$S at levels directly toxic and damaging to the ozone shield, a mechanism implicated in the end-Permian extinction~\cite{KumpPavlovArthur2005}. This is a hard-threshold crossing in the sense of Proposition~\ref{prop:thresholds}: a biogenic residual exceeds processing capacity and triggers collapse. The trigger was external (volcanically forced stratification), consistent with Proposition~\ref{prop:thresholds}(iii), under which the hard threshold is crossed persistently only under forcing; and the biosphere recovered.

\paragraph{A scale mismatch.} Biological aging is substantially a $W$-side process~\cite{LopezOtin2013}, so the individual organism is closer to the hard-threshold, self-poisoning mode, whereas the biosphere's best-attested fate is closer to the resource-exhaustion mode. The cross-scale analogy of Fig.~\ref{fig:scales} is therefore functional rather than mechanistic: the threshold mathematics is shared, but the load-bearing side changes from $W$ at the cellular scale to $R$ at the planetary scale.

\paragraph{An open modelling question.} The two mechanisms run on very different clocks, suggesting a singularly perturbed, two-timescale extension of Eqs.~\eqref{eq:R}--\eqref{eq:N}: a slow, externally forced drift $R_{\max}(t)$ narrowing the margins to $N_R$ and, if processing depends on resources, to $N_c$, until an episode of the Kump--Pavlov--Arthur type ceases to be recoverable. Formalizing this lies beyond the present note; chemostat persistence theory and local-bifurcation analysis of three- and four-dimensional systems (Appendix~\ref{app:context}) are the natural tools.

\section{Relation to existing frameworks and a research programme}
\label{sec:relfra}
None of the ingredients above is new in isolation; the contribution is their deliberately minimal arrangement around finite buffers, reserve times, and processing thresholds.

Relative to the Spaceship Earth metaphor~\cite{Boulding1966,Fuller1969}, the approach pushes toward the accounting discipline of life-support engineering. Relative to planetary boundaries~\cite{Rockstrom2009a,Rockstrom2009b,Steffen2015}, it emphasizes a mechanism by which production and processing rates generate soft and hard thresholds. Relative to \emph{The Limits to Growth}~\cite{Meadows1972}, whose World3 model is far richer empirically, and to System Dynamics~\cite{Forrester1961,Forrester1971}, it trades simulation for a small nonlinear system amenable to closed-form threshold analysis. Relative to ecological-footprint accounting~\cite{Rees1992}, it gives explicit status to harmful residual stocks and processing capacities; relative to industrial ecology and the circular economy~\cite{GraedelAllenby1995,GeissdoerferEtAl2017}, it stresses that recycling shifts thresholds without perfect closure.

The closest comparison is the stock-pollution growth literature. Plourde~\cite{Plourde1972}, Keeler, Spence and Zeckhauser~\cite{Keeler1972}, and Forster~\cite{Forster1973} formalized exploited stocks, pollution accumulation, and abatement; John and Pecchenino~\cite{JohnPecchenino1994} found multiple regimes, and Brock and Taylor~\cite{BrockTaylor2010} connected growth and abatement in the Green Solow model. The soft/hard threshold structure of Proposition~\ref{prop:thresholds} is itself not new: Tahvonen and Withagen~\cite{TahvonenWithagen1996} and Hediger~\cite{Hediger2009} establish closely related and more general results (nonconvex decay, multiple equilibria, irreversibility) in an optimal-control framework. Proposition~\ref{prop:thresholds} should be read as a simpler dynamical restatement of a known family of results.

The differences claimed are modest. Proposition~\ref{prop:thresholds} is a positive threshold and global-stability result from fixed rules rather than an optimization problem, in an engineering-audit vocabulary; different residual classes are treated as different forms of an elimination operator, with the reserve-time diagnostic applied symmetrically to depletion and accumulation; and the organizing analogy is explicitly multiscale, linking cellular homeostasis, physiology, planetary loops, and aging. This is a complementary perspective, not a claim of priority.

A serious planetary audit would require inventory and flux reconstruction for each class, explicit timescale separation, bifurcation analysis of vector-valued models, uncertainty quantification, and data assimilation. Earth-system models with endogenous human dynamics provide part of this infrastructure~\cite{Donges2020}, and reinforcement learning has been explored for sustainable management~\cite{Strnad2019,RuddJones2024,Wolf2023}, although any control framework should remain constrained by conservation laws, transparent risk criteria, and interpretable objectives.

\section{Limitations}

The model is intentionally minimal and is not a calibrated prediction of planetary collapse. The aggregate variables $R$, $W$, and $N$ hide heterogeneous classes whose reserve times differ by orders of magnitude. The nonlinearities in Eqs.~\eqref{eq:functions}--\eqref{eq:g} are chosen for tractability rather than fitted to data. Spatial transport, inequality, trade, substitution, adaptation, political conflict, and endogenous technological change are omitted. Real thresholds may be uncertain, distributed, delayed, and path dependent, and ethical questions cannot be reduced to dynamical stability.

The model does not establish an a priori ranking between resource failure and waste-processing failure; determining which binds first in a real system is an empirical problem. The global convergence in Proposition~\ref{prop:thresholds} relies on an instantaneous response of growth to waste; delays may destabilize the equilibrium. Finally, the thresholds arise endogenously from continuous dynamics. Geomagnetic storms, asteroid impacts, supervolcanism, or war could instead lower $\varepsilon_{\max}$, reduce $\rho$, or shift thresholds discontinuously; such shocks would require stochastic jump processes or time-dependent parameters.

\section{Conclusion}

This note proposes a minimal dynamical language for materially bounded life-support systems, with Earth as the paradigmatic example. Its main claim is not that one constraint is universally dominant, but that sustainability requires simultaneous accounting of finite resource buffers, finite processing capacities, imperfect recycling, and characteristic reserve times.

Three conclusions follow. First, resource depletion and residual accumulation are structurally comparable balance problems: either can bind, depending on its buffer relative to throughput. Second, a simple saturating-processing model already distinguishes a soft threshold, where the harmful stock suppresses activity while processing remains formally capable, from a hard threshold, where production exceeds maximal processing capacity; with immediate feedback, activity settles at the former, and the latter becomes persistent only when activity is forced or its response delayed. Third, carrying capacity is a moving boundary rather than a fixed number.

The constructive legacy of the Spaceship Earth image is therefore not a slogan of scarcity but a research programme: measure the life-support budget, identify the shortest reserve times and smallest margins, model the coupled loops, and design controls that keep trajectories inside a viable region. Cellular homeostasis, organismal physiology, space engineering, and planetary science share one question: how long can organized throughput remain compatible with finite buffers and finite processing?

\appendix
\section{Historical and Mathematical Context of the Model}
\label{app:context}

Each ingredient of Eqs.~\eqref{eq:R}--\eqref{eq:N} belongs to an established mathematical tradition. The regeneration term $\rho(R)$ is Verhulst's logistic law~\cite{Verhulst1838}. The cross-coupling of stocks through $aN$ and $bN$ continues the coupled nonlinear ODEs of Lotka~\cite{Lotka1925} and Volterra~\cite{Volterra1926}, here with a resource, a residual, and an activity level as the interacting quantities. The processing flux $\varepsilon(W)$ is a rectangular-hyperbolic law that recurs, independently, as Michaelis--Menten enzyme kinetics~\cite{MichaelisMenten1913}, Monod's microbial growth on a limiting substrate~\cite{Monod1949}, and Holling's Type~II functional response~\cite{Holling1959} --- a generic consequence of a single rate-limiting handling step, here applied to capacity-limited removal of an accumulating stock. Lobry~\cite{Lobry2025} traces how these ingredients were progressively mathematized in population dynamics. The pairing of a polluting stock with a finite absorption capacity in economics was discussed in Section~\ref{sec:relfra}.

Structurally, Eqs.~\eqref{eq:R}--\eqref{eq:N} form a three-dimensional autonomous system with one saturating nonlinearity and one bilinear coupling, $N\,g(R,W)$. Read with $N$ a population, $R$ a nutrient, and $W$ dead biomass, it is a three-compartment population model in which $\chi$ recycles detritus into nutrient, a classical configuration in population dynamics. Two bodies of general theory apply. Chemostat theory~\cite{SmithWaltman1995,HarmandEtAl2017} gives existence, persistence, and global-stability results for substrate--consumer systems with exactly this saturating form, including conditions for convergence to an interior equilibrium versus washout. At the local level, Bosi and Desmarchelier characterize local bifurcations of three- and four-dimensional systems through the trace, principal minors, and determinant of the Jacobian~\cite{BosiDesmarchelier2019}, and apply this to a resource/pollution/consumption system in which pollution degrades the resource base~\cite{BosiDesmarchelier2018}, a coupling analogous to Eq.~\eqref{eq:rhoW}. These tools would characterize the full three-dimensional dynamics, beyond the two-dimensional waste--activity subsystem of Proposition~\ref{prop:thresholds}; that analysis is not undertaken here, and none of the cited results is relied upon elsewhere in this note.

\section*{Acknowledgements}
I am grateful to Marie-Madeleine Opałka for bringing the Obomsawin quotation to my attention. I thank Romain Gazeau for drawing my attention to soil and water contamination and its consequences for biodiversity, and for suggesting the Amazonian dark earths (\emph{terra preta}) as a historical example of sophisticated waste recycling; Fanny Pratt for pointing me toward the steady-state-economy literature; and Jeanne Gauthier for her attentive engagement with these ideas on waste during a discussion in Montreal in August 2019. I also thank Pedro Rocha (UFRJ) for drawing my attention to the works of Beniger and Tainter, and Claude Lobry for his careful reading of an earlier version, for pointing out the kinship of the model with three-compartment population-dynamics models with recycling, and for drawing my attention to viability theory, and Alain Rapaport (INRAE) for his remarks on the structure of the model and for pointing out the need for a terminal viability condition in the finite-horizon definition of carrying capacity.
I am also grateful to Robert Krueger, Michael Radzicki, Khalid Saeed and Romain Murenzi of Worcester Polytechnic Institute for stimulating discussions during my visit in September 2026, and for introducing me to the system dynamics tradition founded by Jay W. Forrester and developed by his collaborators.

\section*{Author Contributions}
J.-P.~G. conceived the framework, derived the analytical results (with the assistance described in the Declaration on the Use of Generative AI), designed the numerical example, and wrote the manuscript.


\section*{Declaration on the Use of Generative AI}

AI language models from Anthropic (Claude) and OpenAI (ChatGPT)
were used in preparing and revising this article.

Claude Sonnet 4.5 (accessed July 2026) and Claude Sonnet 5
(accessed August 2026), via claude.ai, were used to assist in
condensing and restructuring an earlier, longer working version
of this note and in drafting a cover letter. Claude Opus 5.5
(accessed 6 October 2026), via claude.ai, was used to condense
the present revision to the Concepts and Perspectives length
limit, to draft the Non-Technical and Technical Summaries and
the Social Media Summary, and to draft the cover letter.

ChatGPT (OpenAI, accessed October 2026), via chatgpt.com,
was used to assist in critically examining the mathematical
analysis and its interpretation, identifying consistency issues
concerning the soft and hard thresholds, and revising the
associated discussion and numerical example. It also assisted
in improving wording and conceptual clarity, preparing
annotated summaries of the cited literature, and updating
those summaries to match the revised bibliography.


\section*{Financial Support}
This research received no specific grant from any funding agency, commercial or not-for-profit sectors.

\section*{Conflicts of Interest}
The author declares none.

\section*{Data Availability Statement}
This is a theoretical and conceptual note. No new empirical data were generated or analysed. The toy numerical example is fully specified by the parameter values given in the text and can be reproduced from the closed-form expressions above.


\begin{thebibliography}{99}

\bibitem{Aubin1991} Aubin, J.-P. (1991). \emph{Viability Theory}. Birkh\"auser.

\bibitem{AubinBayenSaintPierre2011} Aubin, J.-P., Bayen, A.~M., \& Saint-Pierre, P. (2011). \emph{Viability Theory: New Directions}. Springer.

\bibitem{Barrangou2007} Barrangou, R., Fremaux, C., Deveau, H., Richards, M., Boyaval, P., Moineau, S., Romero, D.~A., \& Horvath, P. (2007). CRISPR provides acquired resistance against viruses in prokaryotes. \emph{Science}, 315(5819), 1709--1712.

\bibitem{Beniger1986} Beniger, J.~R. (1986). \emph{The Control Revolution: Technological and Economic Origins of the Information Society}. Harvard University Press.

\bibitem{BosiDesmarchelier2018} Bosi, S., \& Desmarchelier, D. (2018). Natural cycles and pollution. \emph{Mathematical Social Sciences}, 96, 10--20.

\bibitem{BosiDesmarchelier2019} Bosi, S., \& Desmarchelier, D. (2019). Local bifurcations of three and four-dimensional systems: a tractable characterization with economic applications. \emph{Mathematical Social Sciences}, 97, 38--50.

\bibitem{Boulding1966} Boulding, K.~E. (1966). The economics of the coming spaceship Earth. In H. Jarrett (Ed.), \emph{Environmental Quality in a Growing Economy}. Johns Hopkins University Press.

\bibitem{BrockTaylor2010} Brock, W.~A., \& Taylor, M.~S. (2010). The Green Solow model. \emph{Journal of Economic Growth}, 15(2), 127--153.

\bibitem{CaldeiraKasting1992} Caldeira, K., \& Kasting, J.~F. (1992). The life span of the biosphere revisited. \emph{Nature}, 360, 721--723.

\bibitem{Carter1983} Carter, B. (1983). The anthropic principle and its implications for biological evolution. \emph{Philosophical Transactions of the Royal Society of London A}, 310(1512), 347--363.

\bibitem{Daly1977} Daly, H.~E. (1977). \emph{Steady-State Economics}. W. H. Freeman.

\bibitem{Danovaro2008} Danovaro, R., Bongiorni, L., Corinaldesi, C., Giovannelli, D., Damiani, E., Astolfi, P., Greci, L., \& Pusceddu, A. (2008). Sunscreens cause coral bleaching by promoting viral infections. \emph{Environmental Health Perspectives}, 116(4), 441--447.

\bibitem{DaughtonTernes1999} Daughton, C.~G., \& Ternes, T.~A. (1999). Pharmaceuticals and personal care products in the environment: agents of subtle change? \emph{Environmental Health Perspectives}, 107(Suppl.~6), 907--938.

\bibitem{Donges2020} Donges, J.~F., Heitzig, J., Barfuss, W., et al. (2020). Earth system modeling with endogenous and dynamic human societies: the copan:CORE open World--Earth modeling framework. \emph{Earth System Dynamics}, 11(2), 395--413.

\bibitem{England2013} England, J.~L. (2013). Statistical physics of self-replication. \emph{Journal of Chemical Physics}, 139(12), 121923.

\bibitem{England2015} England, J.~L. (2015). Dissipative adaptation in driven self-assembly. \emph{Nature Nanotechnology}, 10(11), 919--923.

\bibitem{Forrester1961} Forrester, J.~W. (1961). \emph{Industrial Dynamics}. MIT Press.

\bibitem{Forrester1971} Forrester, J.~W. (1971). \emph{World Dynamics}. Wright-Allen Press.

\bibitem{Forster1973} Forster, B.~A. (1973). Optimal consumption planning in a polluted environment. \emph{Economic Record}, 49(4), 534--545.

\bibitem{FranckBounamaVonBloh2006} Franck, S., Bounama, C., \& von Bloh, W. (2006). Causes and timing of future biosphere extinctions. \emph{Biogeosciences}, 3(1), 85--92.

\bibitem{Fuller1969} Fuller, R.~B. (1969). \emph{Operating Manual for Spaceship Earth}. Southern Illinois University Press.

\bibitem{GeissdoerferEtAl2017}
Geissdoerfer, M., Savaget, P., Bocken, N. M. P., \& Hultink, E. J. (2017).
The circular economy---A new sustainability paradigm?
\textit{Journal of Cleaner Production}, \textbf{143}, 757--768.

\bibitem{GeorgescuRoegen1971} Georgescu-Roegen, N. (1971). \emph{The Entropy Law and the Economic Process}. Harvard University Press.

\bibitem{Glaser2007} Glaser, B. (2007). Prehistorically modified soils of central Amazonia: a model for sustainable agriculture in the twenty-first century. \emph{Philosophical Transactions of the Royal Society B}, 362(1478), 187--196.

\bibitem{GraedelAllenby1995} Graedel, T.~E., \& Allenby, B.~R. (1995). \emph{Industrial Ecology}. Prentice Hall.

\bibitem{HarmandEtAl2017} Harmand, J., Lobry, C., Rapaport, A., \& Sari, T. (2017). \emph{The Chemostat: Mathematical Theory of Microorganism Cultures}. ISTE--Wiley.

\bibitem{Hediger2009} Hediger, W. (2009). Sustainable development with stock pollution. \emph{Environment and Development Economics}, 14(6), 759--780.

\bibitem{Holling1959} Holling, C.~S. (1959). Some characteristics of simple types of predation and parasitism. \emph{The Canadian Entomologist}, 91(7), 385--398.

\bibitem{IAEA2009} International Atomic Energy Agency. (2009). \emph{Classification of Radioactive Waste} (General Safety Guide GSG-1). IAEA.

\bibitem{IPCC2021} IPCC. (2021). \emph{Climate Change 2021: The Physical Science Basis}. Cambridge University Press.

\bibitem{JohnPecchenino1994} John, A., \& Pecchenino, R. (1994). An overlapping generations model of growth and the environment. \emph{The Economic Journal}, 104(427), 1393--1410.

\bibitem{Joyce1994} Joyce, G.~F. (1994). Foreword. In D.~W. Deamer \& G.~R. Fleischaker (Eds.), \emph{Origins of Life: The Central Concepts}. Jones and Bartlett.

\bibitem{Keeler1972} Keeler, E., Spence, M., \& Zeckhauser, R. (1972). The optimal control of pollution. \emph{Journal of Economic Theory}, 4(1), 19--34.

\bibitem{Kidd2007} Kidd, K.~A., Blanchfield, P.~J., Mills, K.~H., Palace, V.~P., Evans, R.~E., Lazorchak, J.~M., \& Flick, R.~W. (2007). Collapse of a fish population after exposure to a synthetic estrogen. \emph{Proceedings of the National Academy of Sciences}, 104(21), 8897--8901.

\bibitem{Kipping2020} Kipping, D. (2020). An objective Bayesian analysis of life's early start and our late arrival. \emph{Proceedings of the National Academy of Sciences}, 117(20), 11995--12003.

\bibitem{Kleidon2010} Kleidon, A., Malhi, Y., \& Cox, P.~M. (2010). Maximum entropy production in environmental and ecological systems. \emph{Philosophical Transactions of the Royal Society B}, 365(1545), 1297--1302.

\bibitem{KooninSenkevichDolja2006} Koonin, E.~V., Senkevich, T.~G., \& Dolja, V.~V. (2006). The ancient Virus World and evolution of cells. \emph{Biology Direct}, 1, 29.

\bibitem{KumpPavlovArthur2005} Kump, L.~R., Pavlov, A., \& Arthur, M.~A. (2005). Massive release of hydrogen sulfide to the surface ocean and atmosphere during intervals of oceanic anoxia. \emph{Geology}, 33(5), 397--400.

\bibitem{LarssonFlach2022} Larsson, D.~G.~J., \& Flach, C.-F. (2022). Antibiotic resistance in the environment. \emph{Nature Reviews Microbiology}, 20(5), 257--269.

\bibitem{Lobry2025} Lobry, C. (2025). \'Evolution du principe d'exclusion comp\'etitive : le r\^ole des math\'ematiques. arXiv:2510.03158 [q-bio.PE].

\bibitem{LopezOtin2013} L\'opez-Ot\'in, C., Blasco, M.~A., Partridge, L., Serrano, M., \& Kroemer, G. (2013). The hallmarks of aging. \emph{Cell}, 153(6), 1194--1217.

\bibitem{LaSalle1960} LaSalle, J.~P. (1960). Some extensions of Liapunov's second method. \emph{IRE Transactions on Circuit Theory}, 7(4), 520--527.

\bibitem{Lotka1925} Lotka, A.~J. (1925). \emph{Elements of Physical Biology}. Williams \& Wilkins.

\bibitem{Meadows1972} Meadows, D.~H., Meadows, D.~L., Randers, J., \& Behrens, W.~W. III. (1972). \emph{The Limits to Growth}. Universe Books.

\bibitem{Michaelian2011} Michaelian, K. (2011). Thermodynamic dissipation theory for the origin of life. \emph{Earth System Dynamics}, 2(1), 37--51.

\bibitem{MichaelisMenten1913} Michaelis, L., \& Menten, M.~L. (1913). Die Kinetik der Invertinwirkung. \emph{Biochemische Zeitschrift}, 49, 333--369.

\bibitem{Mill2011} Mill, J.~S. (2011). \emph{Principles of Political Economy}. Cambridge University Press reprint. Original work published 1848.

\bibitem{Monod1949} Monod, J. (1949). The growth of bacterial cultures. \emph{Annual Review of Microbiology}, 3, 371--394.

\bibitem{Nelson1996} Nelson, M., Dempster, W.~F., Alvarez-Romo, N., \& MacCallum, T. (1996). The legacy of Biosphere 2 for the study of biospherics and closed ecological systems. \emph{Advances in Space Research}, 18(1--2), 179--194.

\bibitem{Oaks2004} Oaks, J.~L., Gilbert, M., Virani, M.~Z., et al. (2004). Diclofenac residues as the cause of vulture population decline in Pakistan. \emph{Nature}, 427, 630--633.

\bibitem{Obomsawin1972} Obomsawin, A. (1972). Quoted in T. Poole, “Conversations with North American Indians,” in R. Osborne (Ed.), Who Is the Chairman of This Meeting? A Collection of Essays, Neewin Publishing, Toronto, p. 43.

\bibitem{OMalleyJames2013} O'Malley-James, J.~T., Greaves, J.~S., Raven, J.~A., \& Cockell, C.~S. (2013). Swansong biospheres: refuges for life and novel microbial biospheres on terrestrial planets near the end of their habitable lifetimes. \emph{International Journal of Astrobiology}, 12(2), 99--112.

\bibitem{Persson2022} Persson, L., Carney Almroth, B.~M., Collins, C.~D., et al. (2022). Outside the safe operating space of the planetary boundary for novel entities. \emph{Environmental Science \& Technology}, 56(3), 1510--1521.

\bibitem{Plourde1972} Plourde, C.~G. (1972). A model of waste accumulation and disposal. \emph{Canadian Journal of Economics}, 5(1), 119--125.

\bibitem{Prigogine1977} Prigogine, I. (1977). Time, structure and fluctuations [Nobel Lecture]. Nobel Foundation.

\bibitem{PrigogineStengers1984} Prigogine, I., \& Stengers, I. (1984). \textit{Order Out of Chaos}, Bantam Books.

\bibitem{Pross2011} Pross, A. (2011). Toward a general theory of evolution: extending Darwinian theory to inanimate matter. \emph{Journal of Systems Chemistry}, 2, 1.

\bibitem{Pross2016} Pross, A. (2016). \emph{What Is Life? How Chemistry Becomes Biology}. Oxford University Press.

\bibitem{RaoultForterre2008} Raoult, D., \& Forterre, P. (2008). Redefining viruses: lessons from Mimivirus. \emph{Nature Reviews Microbiology}, 6, 315--319.

\bibitem{Rees1992} Rees, W.~E. (1992). Ecological footprints and appropriated carrying capacity: what urban economics leaves out. \emph{Environment and Urbanization}, 4(2), 121--130.

\bibitem{Rockstrom2009a} Rockstr\"om, J., Steffen, W., Noone, K., et al. (2009). A safe operating space for humanity. \emph{Nature}, 461, 472--475.

\bibitem{Rockstrom2009b} Rockstr\"om, J., Steffen, W., Noone, K., et al. (2009). Planetary boundaries: exploring the safe operating space for humanity. \emph{Ecology and Society}, 14(2), Article 32.

\bibitem{RuddJones2024} Rudd-Jones, J., Thendean, F., \& P\'erez-Ortiz, M. (2024). Crafting desirable climate trajectories with reinforcement learning explored socio-environmental simulations. arXiv:2410.07287.

\bibitem{Scheffer2009}  Scheffer, M. (2009).
\textit{Critical Transitions in Nature and Society}. Princeton University Press.

\bibitem{SchneiderKay1994} Schneider, E.~D., \& Kay, J.~J. (1994). Complexity and thermodynamics: towards a new ecology. \emph{Futures}, 26(6), 626--647.

\bibitem{Schrodinger1944} Schr\"odinger, E. (1944). \emph{What Is Life?}. Cambridge University Press.

\bibitem{Smil2017} Smil, V. (2017). \emph{Energy and Civilization: A History}. MIT Press.

\bibitem{SmithWaltman1995} Smith, H.~L., \& Waltman, P. (1995). \emph{The Theory of the Chemostat: Dynamics of Microbial Competition}. Cambridge University Press.

\bibitem{SnyderBeattie2021} Snyder-Beattie, A.~E., Sandberg, A., Drexler, K.~E., \& Bonsall, M.~B. (2021). The timing of evolutionary transitions suggests intelligent life is rare. \emph{Astrobiology}, 21(3), 265--278.

\bibitem{SpiegelTurner2012} Spiegel, D.~S., \& Turner, E.~L. (2012). Bayesian analysis of the astrobiological implications of life's early emergence on Earth. \emph{Proceedings of the National Academy of Sciences}, 109(2), 395--400.

\bibitem{Steffen2015} Steffen, W., Richardson, K., Rockstr\"om, J., et al. (2015). Planetary boundaries: guiding human development on a changing planet. \emph{Science}, 347(6223), 1259855.

\bibitem{Steinfeld2006} Steinfeld, H., Gerber, P., Wassenaar, T., Castel, V., Rosales, M., \& de Haan, C. (2006). \emph{Livestock's Long Shadow: Environmental Issues and Options}. FAO.

\bibitem{Strnad2019} Strnad, F., Barfuss, W., Donges, J.~F., \& Heitzig, J. (2019). Deep reinforcement learning in World--Earth system models to discover sustainable management strategies. \emph{Chaos}, 29(12), 123122.

\bibitem{Struik1978} Struik, L.~C.~E. (1978). \emph{Physical Aging in Amorphous Polymers and Other Materials}. Elsevier Scientific Publishing Company.

\bibitem{Suttle2007} Suttle, C.~A. (2007). Marine viruses -- major players in the global ecosystem. \emph{Nature Reviews Microbiology}, 5(10), 801--812.

\bibitem{TahvonenWithagen1996} Tahvonen, O., \& Withagen, C. (1996). Optimality of irreversible pollution accumulation. \emph{Journal of Economic Dynamics and Control}, 20(9--10), 1775--1795.

\bibitem{Tainter1988} Tainter, J.~A. (1988). \emph{The Collapse of Complex Societies}. Cambridge University Press.

\bibitem{UlanowiczHannon1987} Ulanowicz, R.~E., \& Hannon, B.~M. (1987). Life and the production of entropy. \emph{Proceedings of the Royal Society of London B}, 232(1267), 181--192.

\bibitem{VanValen1973} Van Valen, L. (1973). A new evolutionary law. \emph{Evolutionary Theory}, 1, 1--30.

\bibitem{VasasSzathmarySantos2010} Vasas, V., Szathm\'ary, E., \& Santos, M. (2010). Lack of evolvability in self-sustaining autocatalytic networks constrains metabolism-first scenarios for the origin of life. \emph{Proceedings of the National Academy of Sciences}, 107(4), 1470--1475.

\bibitem{Verhulst1838} Verhulst, P.~F. (1838). Notice sur la loi que la population suit dans son accroissement. \emph{Correspondance math\'ematique et physique}, 10, 113--121.

\bibitem{Vernadsky1998} Vernadsky, V.~I. (1998). \emph{The Biosphere}. Copernicus/Springer. English translation of the 1926 original.

\bibitem{Volterra1926} Volterra, V. (1926). Variazioni e fluttuazioni del numero d'individui in specie animali conviventi. \emph{Memorie della Regia Accademia Nazionale dei Lincei}, 2, 31--113.

\bibitem{WalkerDavies2013} Walker, S.~I., \& Davies, P.~C.~W. (2013). The algorithmic origins of life. \emph{Journal of the Royal Society Interface}, 10(79), 20120869.

\bibitem{Watson2008} Watson, A.~J. (2008). Implications of an anthropic model of evolution for emergence of complex life and intelligence. \emph{Astrobiology}, 8(1), 175--185.

\bibitem{Wilkinson2022} Wilkinson, J.~L., Boxall, A.~B.~A., Kolpin, D.~W., et al. (2022). Pharmaceutical pollution of the world's rivers. \emph{Proceedings of the National Academy of Sciences}, 119(8), e2113947119.

\bibitem{Wolf2023} Wolf, T., Nardelli, N., Shawe-Taylor, J., \& P\'erez Ortiz, M. (2023). Can reinforcement learning support policy makers? A preliminary study with integrated assessment models. arXiv:2312.06527.

\end{thebibliography}
\end{document}